\documentclass[aps,pra,floatfix,amsmath,ams,showpacs,12pt]{revtex4}
\usepackage{graphicx}
\usepackage{amsthm}
\usepackage{dsfont}
\usepackage{ulem}
\usepackage{float}
\usepackage{amsfonts}
\restylefloat{figure}
\begin{document}

\title{Entanglement optimizing mixtures of two-qubit states}
\author{K. V. Shuddhodan}
\email[]{rs_shuddhodan@smail.iitm.ac.in}
\author{M. S. Ramkarthik} 
\email[]{ramkarthik@physics.iitm.ac.in}
\author{Arul Lakshminarayan}
\email[]{arul@physics.iitm.ac.in}
\affiliation{Department of Physics, Indian Institute of Technology Madras,
Chennai,
600036, India}

\newcommand{\newc}{\newcommand}
\newc{\beq}{\begin{equation}}
\newc{\eeq}{\end{equation}}
\newc{\kt}{\rangle}
\newc{\br}{\langle}
\newc{\beqa}{\begin{eqnarray}}
\newc{\eeqa}{\end{eqnarray}}
\newc{\pr}{\prime}
\newc{\longra}{\longrightarrow}
\newc{\ot}{\otimes}
\newc{\rarrow}{\rightarrow}
\newc{\h}{\hat}
\newc{\bom}{\boldmath}
\newc{\btd}{\bigtriangledown}
\newc{\al}{\alpha}
\newc{\be}{\beta}
\newc{\ld}{\lambda}
\newc{\sg}{\sigma}
\newc{\p}{\psi}
\newc{\eps}{\epsilon}
\newc{\om}{\omega}
\newc{\mb}{\mbox}
\newc{\tm}{\times}
\newc{\hu}{\hat{u}}
\newc{\hv}{\hat{v}}
\newc{\hk}{\hat{K}}
\newc{\ra}{\rightarrow}
\newc{\non}{\nonumber}
\newc{\ul}{\underline}
\newc{\hs}{\hspace}
\newc{\longla}{\longleftarrow}
\newc{\ts}{\textstyle}
\newc{\f}{\frac}
\newc{\df}{\dfrac}
\newc{\ovl}{\overline}
\newc{\bc}{\begin{center}}
\newc{\ec}{\end{center}}
\newc{\dg}{\dagger}
\newc{\md}{\mbox{mod}}
\newc{\prq}{\mbox{PR}_q}
\newc{\rt}{\sqrt{2}}
\newc{\pro}{|\p_1 \kt \br \p_1|}
\newc{\prt}{|\p_2 \kt \br \p_2|}
\newc{\calN}{{\cal N}}

\begin{abstract}
Entanglement in incoherent mixtures of pure
states of two qubits  is considered via the concurrence measure. A set of pure
states is optimal if the concurrence  for {\it any} mixture of them is the weighted sum of
the concurrences of the generating states. 
When two or three pure real states are mixed it is shown that  $28.5\%$ and
$5.12\%$ of the cases respectively, are optimal.
Conditions that are obeyed by the pure states generating such optimally
entangled mixtures are derived. For four or more pure states it is shown that there are {\it no}
such sets of real states. The implications of these on superposition of two or more dimerized
states is discussed.  A corollary of these results also show in how many cases rebit concurrence can be the 
same as that of qubit concurrence. 
\end{abstract}
\pacs{03.67.-a, 03.65.Bg, 03.67.Mn}

\maketitle


\section{Introduction}

Entanglement properties of pure and mixed quantum states have been the subject
of 
intense and extensive study in the recent past \cite{Horodecki09}. Of these,
entanglement
in qubits 
or spin-1/2 systems have dominated due to their use as fundamental objects in
quantum
computations.  For an arbitrary
state of two qubits the concurrence measure (or its square, called tangle) 
introduced by Hill and Wootters \cite{Wootters} is simply 
calculable from the density matrix and is a measure of entanglement. To be
precise
the entanglement of formation \cite{Nielsen} is a monotonic function of the
concurrence.
The concurrence measure
has been extensively applied in many physical contexts, for instance in the
study of
quantum
phase transitions \cite{Nature}.  In a collection of qubits,
concurrence measures the entanglement present within any chosen pair. Thus due
to
the monogamy property of entanglement \cite{monogamy} it is reasonable to expect
that
states with large 
multipartite entanglement have low or vanishing concurrence. In fact for a
random state
with more than six qubits the probability that a chosen pair has nonzero
concurrence
is vanishingly small, most of the entanglement is of the multipartite kind
\cite{Kendon02,ScottCaves}.

To elaborate on this property, consider a mixed state of two qubits
\beq
\label{first}
\rho= \sum_{i=1}^k p_i \, |\psi_i \kt \br \psi_i |, \;\; k > 1,
\eeq
where $p_i\ge0$,  $\sum_{i}p_i=1$ and the projectors are arbitrary, in
particular, 
$|\psi_i \kt \br \psi_i |$ need not be orthogonal.  
The convexity of concurrence \cite{Uhlmann00} implies that 
\beq
C(\rho) \le \sum_{i=1}^k p_i C(\psi_i)
\label{concconvex}
\eeq
where $C$ is the concurrence function, an entanglement monotone \cite{Wootters}.
Thus the
maximum 
that $C(\rho)$ can attain is the weighted sum of the concurrence of the extremal
(pure)
states.
If there exists a set $\{ |\psi_1 \kt, \cdots, \psi_k \kt \}$  such that
equality is
obtained in 
Eq.~(\ref{concconvex}) for any arbitrary set of weights $\{p_i\}$, it is
referred to
herein as optimal. However note that such a property will be specific to the
concurrence
measure of entanglement.

For states that are real in the standard basis, it is shown that a very large fraction
of states
made by incoherently
superposing 2 two-qubit states optimize their entanglement. This property is
analyzed in detail in this paper and conditions to be satisfied by the extremal
states
such that the resultant density
matrix is optimal are derived. 
 Any real density matrix can be tested for optimality of its diagonal
decomposition using
the inequalities derived.
These are also generalized beyond two states, and it is shown  that for more
than three
real states,
{\it not one} optimal decomposition exists. 

The relevance to superposed dimers
will be studied in section \ref{sec:dimers}.
The relation between entanglement and superpostion of quantum states is an
interesting one
\cite{Linden}. Indeed superposition of states with a tensor product structure is
necessary
for entanglement,
however of course this is not sufficient. There is a significant  amount of
literature
establishing bounds on various entanglement measures for the superposition in
terms of the
entanglement in the states that are being so superposed
\cite{Nisert,Linden,Chang,Heng,Gour,Song,Osterloh}.


\section{Rank-wise study of optimizing mixtures }
\label{sec:rankwise}

If  two pure and real states $|\p_1\kt$ and $|\p_2\kt$ are 
chosen at random, it is shown in this paper that in $28.5\%$ of cases the resultant entanglement in 
$p |\p_1\kt \br \p_1| + (1-p) |\p_2\kt \br \p_2|$ is the maximum
possible, namely equality holds in Eq.~(\ref{concconvex}). This implies that on superposing 2 two-qubit states, $28.5
\%$ of
the states
will remain entangled optimally, as defined in the introduction. The fraction $0.285$ of optimal pairs is interesting
and strong evidence that it is actually $(\pi-2)/4$ is presented in an Appendix.

The conditions under which such optimality occurs is now obtained, however this is done
in a more general setting in what follows.  In particular, 
extending these results to arbitrary mixtures of three real pure states, one
finds that in
about $5.1\%$ of cases this gives rise to optimally entangled states. It is also
then
shown that for four or more states there is not even {\it one} set of pure real
states,
such that all their mixtures are optimal. These generalizations are of
relevance when
more than 2 two-qubit states are superposed.

The reader is first reminded of the procedure to find
the concurrence in $\rho$, a given state of two qubits \cite{Wootters}.
The spin-flipped state $\tilde{\rho}= \sigma_y \otimes \sigma_y \rho^*
\sigma_y \otimes \sigma_y$ is found, where the complex conjugation is done in
the standard
basis.
Then the matrix $\rho \tilde{\rho}$ is diagonalized and has positive eigenvalues
$\mu_1
\ge \mu_2\ge \mu_3\ge \mu_4$. The concurrence $C(\rho)$ is 
$\mbox{max}(0,\sqrt{\mu_1}-
\sqrt{\mu_2}- \sqrt{\mu_3}- \sqrt{\mu_4})$. 

This somewhat involved definition of the concurrence renders it opaque for 
considerations of optimality. However it is possible to  express the concurrence 
of $\rho=\sum_{i=1}^k p_i |\psi_i \kt \br \psi_i |
$ more explicitly in terms of the the states $|\psi_i\kt$ and the weights $p_i$.
Restrict to the case $k \le 4$, that is the size of the generating set of states
is not larger than the maximum  rank of $\rho$. Note that if any $\rho$ is expressed in its eigenbasis this is not a restriction at all. It is now shown that the eigenvalues of  $\rho \tilde{\rho}$ are the same as that of $r'r'^*$ where 
\beq
\label{define-r}
r'_{ij}=\sqrt{p_ip_j}\, r_{ij},\; \mbox{and}\;r_{ij}=\br \psi_i |\sigma_y
\otimes
\sigma_y|\psi_j^* \kt.
\eeq
Thus rather than using the density matrix directly, the pure states comprising a particular
ensemble are used. Note that $|r_{11}|$ and $|r_{22}|$ are the concurrences, $C(\psi_1)$ and $C(\psi_2)$, of the pure states
$|\psi_1\kt$ and $|\psi_2\kt$ respectively.

For convenience the eigenvalue equation for $\tilde{\rho} \rho$ is considered,
whose
right eigenvectors can be written in the nonorthogonal, sub-normalized basis of
the
extremal states as $| \nu \kt = \sum_i \alpha_i |\psi_i' \kt$, where $ |\psi_i'
\kt=\sqrt{p_i}  |\psi_i \kt$. Also, writing
$\rho=\sum_{i=1}^{k} |\psi_{i}^{'} \kt \br \psi_{i}^{'} |$  and using the fact
that $ \br
\psi_{m}^{'}| \tilde{\rho} \rho | \nu \kt=\mu \br \psi_{m}^{'}| \nu \kt$
results in 
\beq
\label{array}
\sum_{i=1}^k(\sum_{j=1}^k \tilde{\rho'}_{mj} t_{ji}' -\mu t_{mi}') \alpha_i =0,
\;\; 1\le
m \le k,
\eeq
where 
\beq 
t'_{ji}=\sqrt{p_jp_i} \; t_{ji}=\sqrt{p_jp_i} \br \psi_j |\psi_i \kt 
\eeq
 is a matrix of inner products (the Gram matrix)
and 
\beq 
\tilde{\rho'}_{mj}=\sqrt{p_mp_j} \; \tilde{\rho}_{mj}=\sqrt{p_mp_j} \br \psi_m
|\tilde{\rho}|\psi_j\kt.
\eeq
Note also that since, $\rho^{*}= \sum p_{i} | \psi_{i}^{*} \kt \br \psi_{i}^{*} |$,
\beq
\tilde{\rho^{'}}_{mj}= \br \psi_{m}^{'} |\tilde{\rho}|\psi_{j}^{'}\kt =  \br
\psi_{m}^{'}
|\sigma_y \otimes \sigma_y(\sum p_{i} | \psi_{i}^{*} \kt \br
\psi_{i}^{*}|)\sigma_y
\otimes \sigma_y |\psi_{j}^{'}\kt
\eeq 
the following matrix identity is readily derived: $\tilde{\rho'}=r' r'^*= r'
r'^{\dagger}$, where $r$ and $r'$ are defined above in Eq.~(\ref{define-r}).

For the $k$ equations in Eq.(\ref{array}) to have non-trivial solutions
$\det(\Lambda)=0$,
where
\beq
\Lambda_{mi}=(\tilde{\rho^{'}}t')_{mi}-\mu (t')_{mi},
\eeq  
which further implies that
\beq
\det(t')\,\det(r'r'^*-\mu I)=0.
\eeq
As the number of vectors $|\psi_i\kt$ are no larger in number than the
dimensionality of
the Hilbert space, we assume them to be independent
and therefore $\det(t')=p_1\cdots p_k \det(t) \ne 0$. Thus 
\beq 
\det(r'r'^*-\mu I)=0.
\eeq
and hence the characteristic polynomials of $\rho \tilde{\rho}$ and $r'r'^*$ are identical.

 If the state $\rho$ is real in the computational basis then the eigenvalue
problem of $r'
r'^*$ is that of 
$r'^2$, whose eigenvalues are the square of the eigenvalues of $r'$, which are
indicated
as $\ld$. The expression for concurrence is derived and the conditions for
optimality are now considered case-by-case starting from $k=2$.

\subsection{Rank-2 density matrices: ${\mathbf k=2}$}
For the case of mixtures of two real pure states, $k=2$, the above considerations
 lead to the following characteristic equation of  $r'$ (as defined in Eq.~(\ref{define-r})):
\beq
\lambda^2-\xi \lambda + p_1 p_2 \chi_{12}=0
\eeq
where $\xi=p_1 r_{11}+p_2 r_{22}$, and $\chi_{12}=(r_{11} r_{22}-r_{12}^2)$. 
The eigenvalues are therefore 
\beq
\lambda_{\pm}=\frac{1}{2}(\xi \pm \sqrt{\xi^2-4 p_1p_2\chi_{12}})
\eeq

If $\chi_{12}>0$ and $\xi>0$ then $\lambda_{+}>\lambda_{-}>0$, and 
$C(\rho)=\lambda_{+}-\lambda_{-}$.  Alternatively if $\chi_{12}>0$ and $\xi<0$
then
$\lambda_{-}< \lambda_{+}<0$, and  $C(\rho)=|\lambda_{-}|-|\lambda_{+}|$. Thus
if
$\chi_{12}>0$, irrespective of the sign of $\xi$,
we have that 
\beq
C(\rho)=\sqrt{\xi^2-4p_1p_2\chi_{12}}<|\xi|<p_1C(\psi_1)+p_2C(\psi_2),
\eeq
confirming the convexity of concurrence. 

The case $\chi_{12}<0$ is more interesting and leads to optimality. If $\xi>0$
then
$\lambda_{-}<0<\lambda_{+}$ and $\lambda_{+}>|\lambda_{-}|$. It follows that
$C(\rho)=\lambda_{+}-|\lambda_{-}|=\lambda_{+}+\lambda_{-}=\xi$. Combining a
similar
analysis of the case $\xi<0$ one gets that when $\chi_{12}<0$, irrespective of
the sign of
$\xi$
\beq
\label{rebit}
C(\rho)=|\xi|=|p_1r_{11}+p_2 r_{22}|.
\eeq
This brings us to the possibility that if $r_{11}r_{22}>0$ then
$C(\rho)=p_1C(\psi_1)+p_2C(\psi_2).$ Thus when $|\psi_1\kt$ and $|\psi_2\kt$
satisfy the
conditions that $r_{11}r_{22}>0$ and $r_{11}r_{22}-r_{12}^2<0$ {\it any} arbitrary
mixture of
these pure states has the maximum entanglement which is their 
average entanglement. This is the first set of optimality conditions that we derive.

The set of such optimal states is a subset from {\it pairs} of real states. Each real state
of two qubits is characterized by $4$ real coefficients, say $x_i$, $i=1,\ldots,4$. The normalization 
condition means that there is a isomorphism between these and the 3-sphere $x_1^2+x_2^2+x_3^2+x_4^2=1$.
Apart from the fact that states differing by a sign are really the same (thus the space is a projective space)
the states maybe thought of as points in $S^3$. Thus a pair of real states is a point on the manifold $S^3 \times S^3$,
and the set of optimal states forms a subset therein whose fractional volume is of natural interest.
 
Assume that the real states of two qubits are distributed uniformly on $S^3$, namely choose the 
Haar measure.
Equivalently, the probability density of random real pure states \cite{Brody} of two
qubits is
given by 
\beq 
\label{uniforms3}
P(\{x_i\})=\dfrac{1}{\pi^2} \delta\left(\sum_{i=1}^4 x_i^2-1\right ),
\eeq where $x_i$ are the state
components in a generic basis  such as the computational one. 
The fraction of optimal states $f_2$ is then the
following integral on $S^3 \times S^3$, written in terms of the ambient space components $(x_i, y_i$) of $\mathbb{R}^4\times \mathbb{R}^4$:
\beq
\label{eqn:frac.two}
f_2= \int \Theta(r_{11}r_{22}) \Theta(r_{12}^2-r_{11}r_{22}) P(\{x_i\}) P(\{y_i\})
\prod_{i=1}^4 dx_i dy_i.
\eeq
Here $r_{11}= \br \psi_1 |\sigma_y \otimes \sigma_y |\psi_1 \kt =2(x_2x_3-x_1x_4)$, $r_{22}=\br \psi_2 |\sigma_y\otimes \sigma_y |\psi_2 \kt = 2(y_2y_3-y_1y_4)$,
$r_{12}=\br \psi_1 |\sigma_y\otimes \sigma_y |\psi_2 \kt =x_2y_3+y_2x_3-x_1y_4-x_4y_1$, and $\Theta$ is the Heaviside step function that is
$1$ if its
argument is positive and zero otherwise.
An exact evaluation of this integral seems possible and equal to $(\pi-2)/4$ (see the Appendix), however 
it is quite easy to simulate the process by choosing two independent vectors ${x_i}$ and ${y_i}$ distributed according to 
Eq.~(\ref{uniforms3}) and checking to see if the optimality condition is satisfied. The initial choice of vectors is done by 
simply taking $4$ numbers from any zero centered normally distributed set and normalizing them. This procedure gives the 
fraction $f_2$ to be approximately $0.285$, in good agreement with the value of  $(\pi-2)/4$.

\subsection{Rank-3 density matrices: ${\mathbf k=3}$}

Now we consider the general setting of mixing three real pure states,
$k=3$,
which leads to the cubic equation for
 the eigenvalues of $r'$, (defined in Eq.~(\ref{define-r})):
\beq
f(\ld)=\ld^{3}+\xi_{1}\ld^{2}+\xi_{2}\ld+\xi_{3}=0,
\eeq
where the coefficients $\xi_{i}$ are 
\beq
\xi_{1}=-\sum_{i} p_i r_{ii},\;
\xi_{2}=\sum_{i \ne j} p_i p_j (r_{ii}r_{jj}-r_{ij}^{2}),\;
\xi_{3}=-p_1p_2p_3 \, \det(r).
\eeq
Now we state two Lemmas which are key to understanding the nature of the roots
of cubic 
equations and the possibility of optimal states in the case $k=3$.

\newtheorem{lem}{Lemma }
\begin{lem}
 If  $p(x)$ is a cubic in $x$, $p(x)=x^{3}+ax^{2}+bx+c$ with real coefficients,
has real
roots, and is such that $a,b,c <0$ then $p(x)=0$ has two negative roots and one
positive
root 
with the positive root being greater than the other two in modulus.
\end{lem}

\begin{proof}
 Since $c<0$ the product of the roots is positive which implies that either all
roots are
positive or there are two negative and one positive root. However since $b<0$
the
quadratic $p'(x)$ has one positive and one negative root. Hence all the roots of
$p(x)$
cannot be positive since the roots of $p'(x)$ have to lie between the roots of
$p(x)$ by
Cauchy's mean value theorem for differentiable functions. Hence the polynomial
has two
negative and one positive root. Observe that $a<0$
implies the sum of the roots is positive which implies that the positive root
has the
largest modulus.
\end{proof}

\begin{lem}
 If $p(x)$ be a cubic in $x$, $p(x)=x^{3}+ax^{2}+bx+c$ with real coefficients
and has real
roots, and is such that $b<0$ and $a,c>0$ then the cubic has two positive roots
and one
negative root 
with the negative root being greater than the other two in modulus.
\end{lem}
\begin{proof}
The proof of this lemma is on  similar lines to the previous one. \end{proof}

{\it Case 1.} Now suppose we have $\xi_{1},\xi_{3}<0$ then along with the
condition that
$\xi_{2}<0$ we have by Lemma 1 that the characteristic equation has two negative
and one
positive root
and that the positive root has the largest magnitude. Without loss of generality
let us
assume that  $\ld_{1}>0$ then we have the square-roots of the eigenvalues of
$\rho
\tilde{\rho}=(\sigma_y\otimes \sigma_y \, \rho)^2$ are
$\ld_{1},-\ld_{2},-\ld_{3}$. Hence
the concurrence of $\rho$ is
$C(\rho)=\mbox{max}(0,\ld_{1}+\ld_{2}+\ld_{3})=\mbox{max}(0,-\xi_{1})=|\xi_{1}
|$. 
Thus
\beq
C(\rho)=p_{1}r_{11}+p_{2}r_{22}+p_{3}r_{33}.
\eeq
If we have $r_{ii}>0$ for $i=1,2,3$ then clearly $C(\p_{i})=r_{ii}$.
Also if $\xi_{2}<0$ then since $p_{i}$'s can be arbitrary positive reals
bounded by 1 we
have each term $p_{i}p_{j}(r_{ii}r_{jj}-r^{2}_{ij})<0$ which means
$r_{ii}r_{jj}-r^{2}_{ij}<0,i
\neq j$ is a necessary condition.
Thus we have that if 
\beq 
r_{ii}>0, r_{ii}r_{jj}-r^{2}_{ij}<0, i\neq j
\eeq
 and 
\beq
\det(r)=\sum_{cyc
}r_{11}(r_{22}r_{33}-r_{13}^{2})-2(r_{11}r_{22}r_{33}-r_{12}r_{23}r_{31})>0
\eeq
 then any mixture of the triple $\{|\p_{1}\kt,|\p_{2}\kt,|\p_{3}\kt\}$ will be
optimally
entangled, that is 

\beq 
 C(\rho)=\sum_{i=1}^3 p_i C(\psi_i)
\eeq

{\it Case 2.} Similarly suppose we have $\xi_{1},\xi_{3}>0$ then along with the
condition
that $\xi_{2}<0$ we have by Lemma 2 that the characteristic equation has two
positive and
one negative root
and that the negative root has the largest magnitude. Let us assume that 
$\ld_{1}<0$ then
we have that the square-roots of the  eigenvalues of $\rho
\tilde{\rho}=(\sigma_y\otimes
\sigma_y \, \rho)^2$ are $-\ld_{1},\ld_{2},\ld_{3}$. Hence the concurrence of
$\rho$
$C(\rho)=
\mbox{max}(0,-(\ld_{1}+\ld_{2}+\ld_{3}))=\mbox{max}(0,\xi_{1})=\xi_{1}$, since
$\xi_{1}>0$. 
Thus we have if 
\beq 
r_{ii}<0, r_{ii}r_{jj}-r^{2}_{ij}<0, i\neq j
\eeq
 and 
\beq 
\det
(r)=\sum_{cyc}r_{11}(r_{22}r_{33}-r_{13}^{2})-2(r_{11}r_{22}r_{33}-r_{12}r_{23}
r_{31})<0
\eeq
 then again any mixture of the triple $\{|\p_{1}\kt,|\p_{2}\kt,|\p_{3}\kt\}$ is
optimally
entangled.
These conditions that are derived for optimality of entanglement for all
mixtures are both
necessary and
sufficient. They can be compactly stated as,
\beq r_{ii} \det(r)>0,\;\; [r]_{ii}<0, \; i=1,2,3,\eeq
where $[r]_{ii}$ are the principal minors of $r$.

 Once again we can therefore write the 
fraction of triples of pure real states of two qubits whose arbitrary mixtures
are
optimally entangled 
in terms of an integral such as in Eq.~(\ref{eqn:frac.two}), which when
evaluated using a procedure extending the previously described one, 
gives the fraction $f_3=0.0512$ or $5.12\%$. Parenthetically, while it is possible that this
integral in all its complexity evaluates exactly as well, the authors are
fairly confident that they cannot find it. 
It is interesting that indeed optimal
triples are
such that any two of them are also optimal. This has to be the case as the
vanishing of
any one
of the probabilities $p_i$ must still be optimal. Thus when we go from mixtures
of rank-2
to rank-3
we see a drastic drop in the percentage of states that lead to optimal
entanglement.
\subsection{ Rank-4 real density matrices are never optimal: ${\mathbf k \geq 4}$}

Incoherently  superposing four of more pure and real states leads to a qualitatively
 different behavior, as shown below.
 When $k=4$ the rank of the eigenvalue problem for $\rho { \tilde\rho}$ is full,
in the
sense that it is the dimensionality 
 of the Hilbert space. From our discussion above it is clear that we have a
quartic
polynomial 
 whose constant term is  $\det(r)$. For optimality we have to have the
concurrence
evaluating to
 the trace of $r'$, we need to have either one eigenvalue positive and three
negative
eigenvalues,
all of them smaller than the positive one in modulus, or one negative eigenvalue
higher in
modulus than
three positive eigenvalues. In either case this implies that $\det(r)<0$. 
However it is easy to prove that $\det(r)>0$. Indeed,
\beq
r_{ij}=\br \psi_i |\sigma_y \otimes \sigma_y |\psi_j\kt = \sum_{mn} F_{im}\br  m
|\sigma_y
\otimes \sigma_y|n \kt F^T_{nj}.
\eeq
 Here $F_{im}=\br \psi_i |m\kt$, and $|m\kt$ is any real orthogonal basis, for
instance it
could be the computational one. Therefore 
\beq
\det(r)=\det[\sigma_y \otimes \sigma_y ] \det(F)^2=\det(F)^2>0,
\eeq
the final inequality following from the reality of the transformation functions.
 Therefore unlike the rank-deficient cases the sign of $\det(r)$ is always
positive and
this rules out the existence of even one real quadruplet such that any arbitrary
mixture
of these remains optimally entangled.
This obviously implies the non-existence of even one set of real optimal state
for $k>4$. It is necessary to have complex states in the ensemble for optimizing the 
entanglement in this case. 

 Entanglement in real qubits have been studied earlier by
also restricting the Hilbert space to the space of reals, the so-called case of ``rebits"
\cite{RungtaRebits}. In this case minimization of the entanglement is also carried out only over the
real ensembles that are realizations of the
density matrix, unlike in the current paper, where we have used the usual
formula for concurrence.
The rebit formula for concurrence is $|\mbox{tr}(\sigma_y \otimes \sigma_y
\rho)|$ which in terms of the quantities introduced in this paper is $| \sum_i p_i \, r_{ii} | $. In the case 
when $k=2$ this is also the actual
concurrence if we only require that $ \chi_{12} = r_{11}r_{22}-r_{12}^2<0$, relaxing the condition that $r_{11} r_{22}>0$.
 This is true in about $78.5 \%$ ($\pi/4$ fraction) of cases when
$k=2$. Thus stated in terms of rebits, the present work also implies that for $78.5\%$
of pairs of real pure states that are mixed, the rebit entanglement coincides with the usual qubit entanglement. Similar generalizations to the case of triples of states gives about $48.5\%$,
which results from relaxing the conditions that all $r_{ii}$ have the same sign. 
However it follows from the above considerations that when we take a mixture of four (or more) real states, the resultant rebit entanglement is {\it always } suboptimal to the usual concurrence obtained on the full complex Hilbert space.

\section{Dimerized states and optimality}
\label{sec:dimers}

As an application of the study of  optimizing mixtures of two qubits, the 
problem of entanglement sharing in superpositions of  states with a dimerized structure 
is now taken up.
If there are many pure states of  $2N$ qubits such that qubits 1 and 2 are entangled only with each other, 3 and 4 with each other and so on, and each of these pairs are in pure states, then an implication of the previous section is that superposing such ``dimerized" states results in rather robust entanglement, especially if only two or three such states are superposed. On adding more such states, the entanglement in the pairs comes down due to the lack of optimality, and will lead to
more global, or multipartite entanglement. Consider the, in general unnormalized, state:
\beq
\label{dimer}
|\psi\kt = \sum_{i=1}^k a_i  \otimes_{j=1}^N |\phi_j^i \kt;\;\;\; \mbox{where} \;\;\;   |\phi_j^i \kt= 
\alpha_{j}^i \vert00\rangle_{j}+\beta_j^i\vert01\rangle_{j}+\gamma_j^i \vert10\rangle_{j}+\delta_{j}^i \vert11\rangle_j,
\eeq
and $\sum_i |a_i|^2=1$. Thus the state is a superposition of $k$ states, labelled by $i$,  each of which has $N$ pairs of entangled two qubit (normalized) pure states, labelled by $j$.  
No two pairs are entangled with each other.
Such superpositions arise in many context, for example in the Resonating Valence Bond states \cite{Anderson}. However note here that
the ``dimers" superposed are of the same kind,  that is, the entangled pairs of particles are the same. 
Most Hamiltonian systems have some form of  time-reversal (anti-unitary)
symmetry that renders
their eigenstates real.

Henceforth $|\phi_1^i \kt$, the state of the first entangled pair in the $i^{th}$ state is referred to
as $|\psi_i \kt $.
For simplicity consider the case when $k=2$.
Let $a_1= \cos(\theta)$ and $a_2=\sin(\theta)$.  
As the superposed states are not orthogonal,
there is the normalization 
factor $\calN$, for the state $|\psi\kt$ in Eq.~(\ref{dimer}) which is 
\beq
\label{mu1}
\calN=1/\sqrt{1+ \sin 2 \theta \mu_1},\;\;\;
\mbox{where}\;\;\;
\mu_1 =  \displaystyle\prod_{j=1}^N
(\alpha_j^1\alpha_j^2+\beta_j^1\beta_j^2+\gamma_j^1\gamma_j^2+\delta_j^1\delta_j^2).
\end{equation}

The reduced density matrix of any two qubits that are entangled in the original
states $|\psi_j^i\kt$, 
which without loss of generality can be taken as the first two qubits, is
\beq
\label{rdm12}
\rho_{12} 
= \calN^2 \left(\cos^2 \theta |\p_1 \kt \br \p_1| + \sin^2 \theta  |\p_2 \kt \br
\p_2|  + 
\mu_2 \sin \theta \cos \theta \, (\vert\p_1\rangle\langle\p_2\vert + \vert\p_2\rangle \langle \p_1 \vert ) \right).
\eeq
Here $\mu_2$ is defined  as 
\begin{equation}
\label{mu2}
\mu_2 = \displaystyle\prod_{j=2}^N
(\alpha_j^1\alpha_j^2+\beta_j^1\beta_j^2+\gamma_j^1\gamma_j^2+\delta_j^1\delta_j^2).
\end{equation}
In all generality this is all that can be said about $\rho_{12}$, however for most states
 the interference or coherence term is negligible,
due to the typical smallness of $\mu_1$ and $\mu_2$. Thus the approximation
$\rho_{12}\approx \rho$, where
\beq
\label{incoh}
\rho =  \cos^2 \theta |\p_1 \kt \br \p_1| + \sin^2 \theta  |\p_2 \kt \br \p_2|, 
\eeq
 is a good one. 

\begin{figure}
\includegraphics{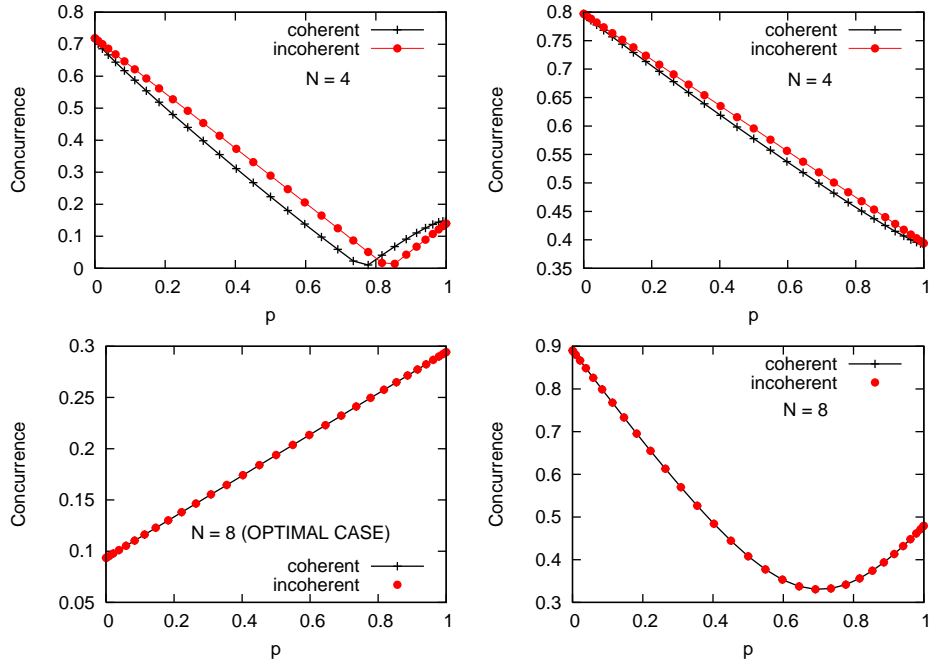}
\caption{Concurrence of the first two qubits {\it vs} $p_1=\cos^2 (\theta)$ for the superposition of
two states ($k=2$) with $N=4$  and $N=8$  (see Eq.~(\ref{dimer})) for two different realizations (left and right columns). The incoherent case is obtained on neglecting the interference terms in the superposition as in Eq.~(\ref{incoh}). The bottom left panel shows a case of optimality. }
\label{coherent}
\end{figure}

To estimate the typical value of $\mu_1$ and $\mu_2$ consider the Hilbert space
of each entangled pair consisting of 
two qubits and take for the distribution of the coefficients the uniform or Haar
measure of Eq.~(\ref{uniforms3}).
The averages of $\mu_1$ and $\mu_2$ are both zero in this ensemble. However the
second
moments are
 nonzero and can be shown to be \[ \br \mu_1^2 \kt = 4^{-N}, \;\;\; \br \mu_2^2
\kt
=4^{-(N-1)}.\]
This follows on observing that $\mu_1$ and $\mu_2$ in Eqs.~(\ref{mu1}, \ref{mu2}) are $N$ and $N-1$ fold products 
of inner-products of two four-vectors.  Consider one such (square of the) inner-product and its ensemble average,
the average being over the distribution where each of the 4-vectors, here denoted simply as $x_i$ and $y_i$ are
distributed according to the measure in Eq.~(\ref{uniforms3}):
\beq
 \left \langle \left( \sum_{i=1}^4 x_i y_i \right)^2 \right \rangle = \sum_{i=1}^4  \left \langle x_i^2 y_i^2 \right \rangle =1/4.
 \eeq
 The first equality follows as the odd powers average to zero, and the second equality follows as the 
 ensemble average of each of the $x_i^2$ is $1/4$, which follows most easily from normalization.
 
Thus $\mu_1$ and $\mu_2$ are typically of the order of $2^{-N}$ and
$2^{-(N-1)}$ respectively. In practice for most states with $N>8$ 
it is difficult to distinguish whether the incoherent two-qubit state $\rho$ obtained on
dropping
the interference term as in Eq.~(\ref{incoh}) is used or the actual reduced density matrix $\rho_{12}$  (Eq.~(\ref{rdm12})) is 
used. This is illustrated in Fig.~(\ref{coherent}), where random realizations of the two-qubit states $|\phi_j^i\kt$ are
used  in Eq.~(\ref{dimer}) with $k=2$. Two such realizations are selected for the cases of $N=4$ and $N=8$ and the concurrence
in the reduced density matrix of the first two qubits are calculated based on the exact state $\rho_{12}$ (Eq.~(\ref{rdm12}), referred 
to in the figure as ``coherent") and the approximation $\rho$ (Eq.~(\ref{incoh}), referred to in the figure as ``incoherent"). The concurrences
are plotted as a function of the mixing between the states that are superposed in Eq.~(\ref{dimer}). It is seen from the figure that when $N=8$ 
whatever difference persist between the entanglement in $\rho$ and $\rho_{12}$ is not visible, while it is for $N=4$. Also from the $N=8$ case,
the realization on the right, illustrates the convexity of concurrence as at $p=0,1$ two different two-qubit pure states are obtained while
for intermediate $p$, the density matrix is an incoherent superposition of these (see Eq.~(\ref{incoh})). However for $N=8$ the realization on the left
is peculiar in that the concurrence for the intermediate values of $p$ is just a linear interpolation of the concurrences of the pure states.
In this case one has hit upon what is studied in this paper as an optimal pair of pure states.

The approximate form of the reduced density matrix in Eq.(\ref{incoh}) 
obtained by using the incoherent superposition condition holds also for complex
states, however optimality results arise in the case of real state as 
studied in detail in the previous section. For $k>2$ similar considerations give rise to states of the form 
as in Eq.~(\ref{first}) where $p_i$ are equal to $|a_i|^2$.
 In terms of dimerized states, we see that superposing two of them leads to pairs of 
 qubits that were originally entangled with each other retaining much of it. About $28.5\%$ of pairs
 of qubits would have simply the weighted entanglements of the states before superposition.
On the other hand superposing three dimerized states leads to
a significant decrease in the fraction of robust dimers, which is now $5.12\%$, and superposing more than three,
the entanglement of a pair is bound to be smaller than the weighted entanglements.

\section{Discussions and Summary}

This work has used a particular measure of entanglement of two-qubit density matrices, namely concurrence, and studied questions of optimality, as defined herein, mainly restricting to the space of real states. Concurrence is only one measure of entanglement, but the entanglement of formation being a monotonic function of it renders it rather unique and as such this measure has been used in very many studies. What exactly the issue of optimality says about the geometry of the quantum space of states \cite{ZyckBook} is an interesting question that is not pursued here. It may also be interesting to study other measures of non-classical correlation, such as discord \cite{ZurekDiscord} from this perspective.

One possible application of optimality was discussed in the superposition of dimerized states. Such states are found in many quantum spin systems, such as the Majumdar Ghosh hamiltonian \cite{MajGhosh}, and superpositions are relevant in the  neighborhood of avoided crossings where it is known that interesting transformations of entanglement occur \cite{KarthikAudityaArul}.
Results not shown here indicate circumstances under which such intra-dimer entanglement can be broken if a dimerized state is superposed with a non-dimerized, completely 
random state. In particular this results in dimer density matrices being very 
close to Werner states \cite{Werner} and therefore results in the entanglement
of dimers vanishing when the random state component is more than $2/3$. In general the effect of superposition
on entanglement has been studied vigorously \cite{Linden}, as well as entanglement in the RVB states has
been explored \cite{RVB}. The results discussed in this paper may add in some small measure to 
the understanding of entanglement in such contexts. 

Real states are often obtained as eigenstates of time-reversal symmetric systems
and the discussion here will be of relevance to such systems, for instance to spin
chains that have time-reversal. 
Clearly when the states are complex the above approach to finding optimality
conditions
does not work. 
In the case when $k=2$ the optimality condition is found to be the same as that 
for real states, namely  $ 0<r_{11} r_{22}/r_{12}^2 <1 $ (except now $r_{12}$ is not
restricted to the reals). 
Thus it is conceivable that as this is the unit  interval in the complex plane,
the measure of optimal states is zero. However of course, there could be an infinity of
these, for
instance
all real states are possible candidates.

\begin{acknowledgments}
We thank Karol Zyczkowski, Steven Tomsovic, Suresh Govindarajan and V. Balakrishnan for discussions. This work was in part supported
by the DST project SR/S2/HEP-012/2009.
\end{acknowledgments}

\appendix*

\section{Evaluation of an integral for the optimal fraction $f_2$} 
\label{appa}

To recall, the integral is 
\beq
\label{theintegral}
f_2= \int \Theta(r_{11}r_{22}) \Theta(r_{12}^2-r_{11}r_{22}) P(\{x_i\}) P(\{y_i\})
\prod_{i=1}^4 dx_i dy_i,
\eeq
where $r_{ij}= \br \psi_i |\sigma_y \otimes \sigma_y |\psi_j \kt$ ,  $|\psi_i\kt$, $=1,2$ are two real and normalized 4-vectors,
and the measures $P$ are the uniform measures in Eq.~(\ref{uniforms3}). 
First, the particular Pauli matrix $\sigma_y$ that appears can be replaced by the other
Pauli matrices. In particular the $\sigma_z$ matrix being diagonal, offers a simpler look. This
replacement is quite easily seen to be equivalent to some $45^{\circ}$ rotations of the original variables.

Also dropping the constraint on the product $r_{11} r_{22}$ will be useful. If the resultant integral
is denoted as $f$, then it is shown below that $f_2$ is simply $f-1/2$.
Next, it is proven that as far as $f$ is concerned, the two 4-vectors can be taken to be orthogonal. Decompose 
say $|\psi_2 \kt$ along the vector $|\psi_1\kt$ and one orthogonal to it:
\beq
|\psi_2\kt = \cos(\theta) |\psi_1\kt + \sin(\theta) |\eta \kt,
\eeq
where $ \br \eta | \psi_1 \kt =0$. No additional phases are involved as the states are all real.
A straightforward calculation shows that 
\beq
r_{12}^2-r_{11}r_{22}=\sin^2(\theta) \left( \br \psi_1 |\sigma_z \otimes \sigma_z |\eta \kt^2 -
\br \eta |\sigma_z \otimes \sigma_z |\eta \kt \br \psi_1 |\sigma_z \otimes \sigma_z |\psi_1 \kt \right).
\eeq
The quantity within the brackets is precisely the same combination as in the L.H.S., except
that instead of $|\psi_i\kt$, the vectors are orthogonal. As $\sin^2(\theta)$ has a constant positive sign,
this proves that we can consider the pairs, to begin with, as being orthogonal. In other words the sign of the 
combination $r_{12}^2-r_{11}r_{22}$ is invariant under the Gram-Schmidt orthogonalization process.

The additional constraint of the vectors being orthonormal introduces an additional Dirac delta function term
in the measure. Writing the fraction $f$ as a ratio $f_n/f_d$, the numerator $f_n$ and the denominator $f_d$ 
are given by 
\beq
\label{numf}
f_n =\dfrac{8}{\pi^2} \int  \Theta \left( r_{12}^2 - r_{11} r_{22} \right) \delta \left( \sum _i x_i y_i \right) \delta \left( \sum _i x_i^2-1 \right) \delta \left( \sum _i  y_i^2-1 \right) dx \, dy, 
\eeq
and 
\beq
\label{denf}
f_d =  \dfrac{8}{\pi^2} \int \delta \left( \sum _i x_i y_i \right) \delta \left( \sum _i x_i^2-1 \right) \delta \left( \sum _i  y_i^2-1 \right) dx \, dy. 
\eeq
All the sums are from $1$ to $4$ and $dx\, dy$ is the Euclidean eight dimensional volume element. The factor $8/\pi^2$ is introduced for later convenience alone. To be further explicit the combination $r_{12}^2 - r_{11} r_{22}=$
\beq
\label{explicit}
\left( x_1 y_1-x_2y_2-x_3 y_3 + x_4 y_4\right)^2 - (x_1^2 -x_2^2-x_3^2 +x_4^2) (y_1^2 -y_2^2-y_3^2 +y_4^2).
\eeq
Introducing a series of transformation to various two-dimensional polar coordinates (in the ($x_1,x_4$) pair,
the $(x_2,x_3)$ pair, etc., as well as in the resulting radii) and performing two delta function integrals
corresponding to the normalizations results is:
\beq
\begin{split}
f_n=2 \int \Theta \left[ \left( \cos(\alpha) \cos(\beta) \cos(\theta) -\sin(\alpha) \sin(\beta) \cos(\phi)\right)^2 -\cos(2 \alpha) \cos(2 \beta) \right] \\ \delta \left[ \cos(\alpha) \cos(\beta) \cos(\theta) +\sin(\alpha) \sin(\beta) \cos(\phi)\right]
\sin(2 \alpha) \sin(2 \beta) d\alpha d \beta d\theta d \phi,
\end{split}
\eeq 
and a corresponding integral for $f_d$, only without the Heaviside theta function constraint.
Here $\alpha, \beta \in [0, \pi/2]$ while $\theta, \phi \in [0, 2 \pi]$. As a check, the integral can be easily done 
without either the Heaviside theta or the Dirac delta functions to give $8 \pi^2$, which is exactly the
factor that follows from the normalization of the two Dirac delta normalization constraints; see Eq.~(\ref{uniforms3}),
if one takes into account the factor $8/\pi^2$ that is introduced in Eq.~(\ref{numf}).

Introducing variables $v$ and $u$ as $\cos(\alpha) \cos(\beta) \cos(\theta) \pm \sin(\alpha) \sin(\beta) \cos(\phi)$
respectively, allows the delta function integration over the $v$ variable to be performed and results in:
\beq
f_n = \int \dfrac{\Theta \left( u^2 - \cos(2 \alpha ) \cos(2 \beta ) \right) \, du \, d\alpha \, d\beta }{ \sqrt{ 1-\dfrac{u^2}{4 \cos^2(\alpha) \cos^2(\beta)}} \sqrt{ 1-\dfrac{u^2}{4 \sin^2(\alpha) \sin^2(\beta)}}}.
\eeq
Notice that the given range of $\alpha$ and $\beta$, $[0, \pi/2]^2$ can be divided into four equal squares,
such that the $\Theta$ function constraint is effective only in $[0,\pi/4]^2$ and in $[\pi/4,\pi/2]^2$, as $\cos(2 \alpha) \cos(2 \beta)$ is negative elsewhere.  The range of the $u$ integration is restricted depending on $\alpha, \beta$.
Taking the range $[0,\pi/4]$, the contribution from it is denoted $f_n^{ll}$, $ll$ for lower-lower, we have:
\beq
f_n^{ll} = \int_{0}^{\pi/4} d \alpha \int_{\beta_0}^{\pi/4} d \beta \int_{u_0}^{u_1} \dfrac{du}{\sqrt{ 1-\dfrac{u^2}{4 c_{ \alpha}^2 c_{ \beta}^2}}\sqrt{ 1-\dfrac{u^2}{4 s_{ \alpha}^2 s_{ \beta}^2}}},\eeq
where
\[ \beta_0 = \sin^{-1} \sqrt{1/2-s_{\alpha}^2}, \; u_0=\sqrt{c_{2 \alpha} c_{2 \beta}}, \;  \mbox{and}\;  u_1= 2 s_{\alpha} s_{\beta}, \]
and $s_{\alpha}$ stands for $\sin(\alpha)$ etc.. The limits of the integration are such that the $\Theta$ function constraint is satisfied as well as the square-roots are real numbers. The denominator fraction can be similarly split up, and in fact $f_d^{ll}$ differs from the above in that $\beta_0=0$ as well as $u_0=0$. It is also not hard to see that $f_n^{uu}$ ($uu$ for upper-upper) is same as $f_n^{ll}$ and similarly for $f_d$. In the other two regions,
as the constraint is not operational, and because of symmetry, it follows that: $f_d^{ul}=f_d^{lu}=f_n^{ul}=f_n^{lu} \equiv f^{ul}$. It then follows that
\beq
f=\df{2 f^{ul}+2 f_n^{ll}}{2 f^{ul}+2 f_d^{ll}}.
\eeq
An evaluation of the three corresponding integrals is carried out numerically and results in 
$f_n^{ll}=0.116850...$, $f^{ul}=1/2$, $f_{d}^{ll}=0.285391...$. It is remarked that standard softwares could not evaluate the integrals symbolically, however the numerical results are sufficient to give the fraction $f=0.785398...$,
which is $\pi/4$ to 1 part in $10^{6}$. It is also easy to similarly see that $f_d^{ll}=(\pi-2)/4$ and from consistency$f_n^{ll}= (\pi^2-8)/16$.

Returning to the integral in Eq.~(\ref{theintegral}) for the fraction $f_2$ we see that there are two $\Theta$ constraints, while we have considered only one, namely the second one.  If the first constraint $\Theta(r_{11} r_{22})$ alone is present, it is easy to see that the integral is $1/2$. This also follows from the fact that $r_{11}$ and $r_{22}$ are both independent and uniformly distributed. We also state parenthetically without proof  that $r_{12}$ is distributed according to the semi-circle distribution. Now, if $r_{11} r_{22}$ is negative, then surely
$r_{12}^2-r_{11}r_{22}$ is positive, which is $50\%$ of the time. Thus the fraction of cases when $r_{11} r_{22}$ is 
positive {\it and} $r_{12}^2-r_{11}r_{22}$ is positive is $f-1/2$, which is precisely the required fraction $f_2$. Thus we have evaluated $f_2$ and presented sufficient evidence that it is actually $(\pi-2)/4$. It is not clear if it is only
a coincidence that this is also precisely $f_d^{ll}$.
The evaluation presented here may, by far, 
not be the ``optimal" one, but is the best the authors could come up with.

\end{document}